%% file: main.tex
\documentclass[conference,letterpaper]{IEEEtran}
\usepackage{todonotes}

\usepackage{amsthm}
\usepackage{amsfonts}

\newcommand{\overbar}[1]{\mkern 3.5mu\overline{\mkern-3.5mu#1\mkern-3.5mu}\mkern 3.5mu}

\newcommand{\delete}{\backslash}
\newcommand{\contract}{/}

\makeatletter
\newcommand*{\centerfloat}{%
	\parindent \z@
	\leftskip \z@ \@plus 1fil \@minus \textwidth
	\rightskip\leftskip
	\parfillskip \z@skip}
\makeatother

\usetikzlibrary{decorations.pathreplacing}
\usetikzlibrary{arrows}
\usetikzlibrary{shapes,snakes}
\usetikzlibrary{calc}
\usetikzlibrary{decorations}
\usetikzlibrary{fadings}
\tikzfading[name=fade out, 
inner color=transparent!0,
outer color=transparent!100]
\pgfdeclaredecoration{lightning bolt}{draw}{
	\state{draw}[width=\pgfdecoratedpathlength]{
		\pgfpathmoveto{\pgfpointorigin}%
		\pgfpathlineto{\pgfpoint{\pgfdecoratedpathlength*0.6}%
			{-\pgfdecoratedpathlength*.1}}%
		\pgfpathlineto{\pgfpoint{\pgfdecoratedpathlength*0.55}{0pt}}%
		\pgfpathlineto{\pgfpoint{\pgfdecoratedpathlength}{0pt}}%
		\pgfpathlineto{\pgfpoint{\pgfdecoratedpathlength*0.4}%
			{\pgfdecoratedpathlength*.1}}%
		\pgfpathlineto{\pgfpoint{\pgfdecoratedpathlength*0.45}{0pt}}%
		\pgfpathclose%
	}%
}

\newcommand{\CZ}{\mathrm{CZ}}
\newtheorem{theorem}{Theorem}
\newtheorem{lemma}{Lemma}
\newtheorem{corollary}{Corollary}
\newtheorem{definition}{Definition}
\newtheorem{prop}{Proposition}
\usepackage[utf8]{inputenc} 
\usepackage[T1]{fontenc}
\usepackage{url}
\usepackage{ifthen}
\usepackage{cite}

\usepackage[cmex10]{amsmath} 

\usepackage{physics}
\usepackage{mathtools}

\newcommand{\kg}[1]{{\textcolor{black}{#1}}}

\begin{document}
	\title{Bipartite entanglement of noisy stabilizer states through the lens of stabilizer codes} 
	


	\author{\IEEEauthorblockN{Kenneth Goodenough}
		\IEEEauthorblockA{Computer Science\\ UMass Amherst}
  \and
  		\IEEEauthorblockN{Aqil Sajjad}
		\IEEEauthorblockA{Optical Sciences\\ Univ. of Arizona}
  \and
  		\IEEEauthorblockN{Eneet Kaur}
		\IEEEauthorblockA{Optical Sciences\\ Univ. of Arizona}
  \IEEEauthorblockA{Cisco Quantum Lab}

  \and
		\IEEEauthorblockN{Saikat Guha}
		\IEEEauthorblockA{Optical Sciences\\ Univ. of Arizona}
  \and
		\IEEEauthorblockN{Don Towsley}
		\IEEEauthorblockA{Computer Science\\ UMass Amherst}
	}


	\maketitle	
	\begin{abstract}
		Stabilizer states are a prime resource for a number of applications in quantum information science, such as secret-sharing and measurement-based quantum computation. This motivates us to study the entanglement of stabilizer states across a bipartition, in the presence of Pauli noise. We show that the spectra of the corresponding reduced states can be expressed in terms of properties of an associated stabilizer code. In particular, this allows us to show that the coherent information is related to the so-called syndrome entropy of the underlying code. We use this viewpoint to find stabilizer states that are resilient against noise, allowing for more robust entanglement distribution in near-term quantum networks. We specialize our results to the case of graph states, where the found connections with stabilizer codes reduces back to classical linear codes for dephasing noise. On our way we provide an alternative proof of the fact that every qubit stabilizer code is equivalent up to single-qubit Clifford gates to a graph code.
	\end{abstract}
	
	\section{Introduction}
	Entanglement forms the core principle behind many applications of quantum communication. The entanglement of stabilizer states has been found to be a resource vital for sensing, secret-sharing and computation~\cite{PhysRevA.52.R2493,bartolucci2021fusionbased,Kitaev2003,Briegel,RevModPhys.83.33,Azuma2015, Damian2020}. Distributing such states is a non-trivial task, and inevitably leads to noise being imparted unto the state. In this work, we are concerned with finding stabilizer states whose entanglement is robust against noise. As a quantifier of the entanglement, we use the coherent information, which is a lower bound on the distillable entanglement~\cite{schumacher1996quantum}.
	
	To find robust states, we first need to be able to describe a stabilizer state after noise has been applied. For Pauli noise, we show that the total state and its marginals are described in terms of stabilizer codes associated with the state. We then specialize our results to uniform depolarizing/dephasing noise, to phrase the resultant spectra in terms of weight enumerators. As we motivate and confirm numerically, the robustness of the state is directly linked to the underlying stabilizer codes.

\begin{theorem}\label{thm:main_informal1}
(Informal) The coherent information (with respect to a bipartition $A, B$ ) of a stabilizer state $\rho_{AB}$ that has undergone noise on each qubit is given by
\begin{align}
\hat{S}(\rho_{B}) - \hat{S}(\rho_{AB}) = k+S(\mathbf{s}_{B})-S(\mathbf{s}_{AB}) \ , 
\end{align}
where $\hat{S}\left(\cdot \right)$ indicates the von Neumann entropy, and $S(\mathbf{s}_{B})$ and $S(\mathbf{s}_{AB})$ are \emph{syndrome entropies}~\cite{gallager1962low, grover2007upper, wadayama2006ensemble, matas2016analysis, li2004slepian} of two stabilizer codes (with parameters $\left[n+k, 0, d'\right]$ and $\left[n, k, d\right]$) associated with $\rho_{AB}$. 
\end{theorem}

\kg{We clarify in the main text how to construct the two associated stabilizer codes.} We show that the above syndrome entropies can be expressed in terms of so-called weight enumerator polynomials for uniform depolarizing noise.  

We also consider the case where only Bob's qubits experience noise, which corresponds to the setting of Alice transmitting one half of her state through a tensor product of quantum channels.

\begin{theorem}\label{thm:main_informal3}
(Informal) The coherent information of a noisy stabilizer state $\rho_{AB}$ where only Bob's qubits have experienced noise is given by
\begin{align}
\hat{S}(\rho_{B}) - \hat{S}(\rho_{AB}) = k+S(\mathbf{s}_{B})-S(\mathbf{s}_{B}^\perp) \ , 
\end{align}
where $S(\mathbf{s}_{B}^\perp)$ is the syndrome entropy of the dual of the code associated to $\mathbf{s}_B$.
\end{theorem}

	We further specialize Theorem \ref{thm:main_informal1} to graph states \kg{(which are stabilizer states that are constructed by applying controlled-$Z$ ($\CZ$)} gates to qubits in the $\ket{+}$ state), where we show how to deal with dephasing and depolarizing noise applied both before and after application of the $\CZ$ gates. For dephasing, we find a link between the robustness of the state and classical codes.

\begin{theorem}\label{thm:main_informal2}
(Informal) The coherent information of a graph state $\rho_{AB}$ that has undergone uniform dephasing with parameter $p$ with respect to a bipartition $A, B$ is given by
\begin{align}
\hat{S}(\rho_{B}) - \hat{S}(\rho_{AB}) = k+S(\mathbf{s}_{B})-\left(n+k\right)\overbar{S}\left(p\right) \ , \nonumber
\end{align}
 where $n+k$ equals the total number of qubits, $S(\cdot)$ is the syndrome entropy of a classical code whose generator matrix is given by the biadjacency matrix of the bipartite graph consisting of edges between $A$ and $B$, and $\overbar{S}\left(\cdot \right)$ is the binary entropy.
\end{theorem}
Theorem \ref{thm:main_informal2} straightforwardly generalizes to the case where only Bob's qubits experience decoherence.

 \kg{The above two results aid in calculating the coherent information of noisy stabilizer states, which is more tractable than the naive approach of diagonalizing the density matrices directly. Furthermore, since the coherent information can be phrased entirely in terms of properties of the associated code, a first intuition would be that graphs with good associated codes provide robustness against noise, at least in the low noise regime. We confirm this intuition numerically in the main text.}

\kg{Our results follow from interpreting stabilizer codes as arising from \emph{contraction} and \emph{deletion} operations on stabilizer states. We use this perspective to give alternative proofs of 1) that every stabilizer code is equivalent to a graph code, 2) a graph-theoretic description of equivalent graph codes. }
 
	\section{Preliminaries}
	We use most of the standard notation used for graphs, classical codes, stabilizer states and graph states. The set of all Pauli strings on $n$ qubits is denoted by $\mathcal{P}_n$.
 A commutative subgroup $H$ of $\mathcal{P}_n$ that does not contain $-I$ is called a stabilizer group. Associated with each such $H$ with $|H|= 2^{n-k}$ is an $\left[n, k, d\right]$ code. That is, it is the unique subspace that is point-wise fixed by each $h\in H$. We use the term code to refer to both subspaces and their associated stabilizer groups. In this paper, we ignore phases of Pauli strings. \kg{Each stabilizer group also has a group of logical operators, $H_{\textrm{logical}}$. We use the convention that this group consists of all elements in $\mathcal{P}_n$ which commute with all elements in $H$.} We will use the overline notation to indicate that a stabilizer group corresponds to a stabilizer state.

\kg{Given a generating set of a stabilizer group $H$, a syndrome is often defined as an observed measurement pattern when measuring the elements from the generating set. Note that this definition depends on an arbitrary choice of the generating set. A more abstract definition exists for syndromes that does not depend on such a choice. Namely, syndromes are identified as the cosets $P\cdot H_{\textrm{logical}}$ for $P\in \mathcal{P}_n$. The first definition can be recovered by noting that two elements are in the same coset iff they have the same (anti-)commutation relations with any generating set.}  

 A graph state can be defined as the state reached after applying $\CZ$ gates according to the edges of a graph $G$ to $\ket{+}^{\otimes n}$. Since $\CZ$ gates commute, this is well-defined. \kg{A graph state can also be defined by a natural generating set of its stabilizer group, which is given by $X_v \prod_{w\in N_v}Z_w$ for each vertex $v\in V$~\cite{hein2006entanglement}.}

	Each element $P\in\mathcal{P}_n$ has a \emph{weight} $\textrm{wt}\left(P\right)$, i.e.~$\textrm{wt}\left(P\right)$ counts the number of non-identity elements in $P$. For an arbitrary subset $K\subseteq \mathcal{P}_n$, define $\mathcal{W}_w\hspace{-0.5mm}\left(K\right)$ to be the number of elements in $K$ with weight equal to $w$. The \emph{weight enumerator polynomial of $K$} is defined as $\mathcal{W}\hspace{-0.5mm}\left(K, x, y\right) = \sum_{w=0}^n \mathcal{W}_w\hspace{-0.5mm}\left(K\right)x^{n-w}y^w$~\cite{gottesman1997stabilizer}.

	The depolarizing channel on a single qubit is given by $D(\cdot) = \lambda\left(\cdot \right)+\left(1-\lambda\right)\textrm{Tr}\left(\cdot\right) \frac{I}{2}$ for some $\lambda\in \left[0, 1\right]$\footnote{Although we do not consider it here, technically $\lambda$ can take values outside of this range while still leading to a valid quantum channel~\cite{king2003capacity}.}. Alternatively, with probability $\lambda$ no noise is applied, while with probability $1-\lambda$ a single-qubit Pauli $I, X, Y, Z$ is applied with uniform probability. Generalizing to uniform depolarizing noise on $n$ qubits, we find that a Pauli string $P$ with weight $\textrm{wt}(P)$ is applied with probability $\left(\lambda+\frac{1-\lambda}{4}\right)^{n-\textrm{wt}\left(P\right)}\left(\frac{1-\lambda}{4}\right)^{\textrm{wt}\left(P\right)} = \frac{1}{4^n}\left(1+3\lambda\right)^{n-\textrm{wt}\left(P\right)}\left(1-\lambda\right)^{\textrm{wt}\left(P\right)} $.
		
	We are interested in bipartite entanglement between Alice and Bob, where we assume that Alice holds $k$ qubits and Bob holds $n$ qubits. This convention is convenient for relating properties to $\left[n, k, d\right]$ codes. Finally, we use the shorthand notation $U(A)\equiv UAU^\dagger$.
	
\subsection{Contraction and deletion}
        In this section we present two important operations, \emph{contraction} and \emph{deletion}.
        
	\begin{definition}\label{def:deletion}
		The deletion of a subset of qubits $T$ of a stabilizer group $H$ yields a new group $H\delete T$. The group $H\delete T$ has the same elements as $H$, but with the entries corresponding to $T$ removed. We also use $h\delete T$ when deleting entries only from a single element $h\in H$.
	\end{definition}
	
	\begin{definition}\label{def:contraction}
		The contraction of a subset of qubits $T$ of a stabilizer group $H$ yields a new group $H\contract T$. The group $H\contract T$ consists of those elements of $H$ that do not have support on $T$. Furthermore, the entries corresponding to $T$ are removed.
	\end{definition}

We show an example of the contraction and deletion operations for the complete graph state on three vertices in Fig.~\ref{fig:contraction_deletion_example}. Note that we use the shorthand $H\delete t$ ($H\contract t$) when deleting (contracting) a singleton set $\lbrace{t\rbrace}$.

	Deletions and contractions are closely connected through the \emph{dual} operation.
	
	\begin{definition}
		The dual of a subgroup $H$ of $\mathcal{P}_n$ (which need not be a stabilizer group) is the largest subgroup $H^\perp$ of $\mathcal{P}_n$ that commutes elementwise with all elements of $H$.
	\end{definition}
	
	The dual of the stabilizer group of a code is exactly the group of logical operators of that code\kg{, i.e. $H^\perp = H_{\textrm{logical}}$}~\cite{gottesman1997stabilizer}. Furthermore, the dual operation is an involution. Contraction and deletion are related through the dual operation as follows
	
	\begin{align}
		(H\delete T)^\perp &= \left(H^\perp\right)\contract T\ ,\label{eq:dual1}\\
		\left(H\contract T\right)^\perp &= \left(H^\perp\right)\delete T \ \label{eq:dual2}.
	\end{align}

	\begin{figure}
		\centerfloat
		\include{contraction_deletion}
		\vspace{-15mm}
		\caption{Example of the contraction and deletion operation on a single qubit of a stabilizer state. We choose the complete graph state $\ket{K_3}$, which has stabilizer group $\overbar{H} = \lbrace{III, XZZ, ZXZ, ZZX, YYI, YIY, IYY, XXX\rbrace}$. Contracting the first qubit yields the code $\overbar{H}\contract 1=\lbrace{II, YY \rbrace}$, while deletion yields $\overbar{H}\delete 1 = \lbrace{II, ZZ, XZ, ZX, YI, IY, YY, XX\rbrace}$. Note that $\overbar{H}\delete 1$ is the group of logical operators of the stabilizer code $\overbar{H}\contract 1$. Furthermore, the \emph{type} of logical operator of an element $h\delete 1$ is given by $h\delete \lbrace 2, 3\rbrace$, e.g.~$ZZX\delete 1 = ZX$ implements a logical $Z$ since $ZZX\delete \lbrace 2, 3\rbrace = Z$.}
  \label{fig:contraction_deletion_example}
	\end{figure}

	The contraction and deletion perspectives are used to reason about marginals of stabilizer states, and especially to relate them to stabilizer codes. As an example of this perspective, note that contracting a subset of qubits of a stabilizer state yields a stabilizer code. In fact, we show that every stabilizer code arises in this fashion.

 \begin{lemma}\label{lemma:surjectivity_lemma}
 For every stabilizer group $H$ on a set of qubits $V$ there exists a stabilizer group $\overbar{H}$ of a stabilizer state on a set of qubits $V' \cup V$ such that $H = \overbar{H}\contract V'$.
 \end{lemma}
 \begin{proof}
See Appendix~\ref{sec:surjectivity}.
 \end{proof}

Given such an $\overbar{H}$ for a stabilizer group $H=\overbar{H}\contract V'$, we find through the duality relations in Eq.~\eqref{eq:dual1} and~\eqref{eq:dual2} that the group of logical operators is given by $H^\perp = \overbar{H}\delete V'$. Furthermore, there is a canonical way to assign the \emph{type} of logical operator to the element $h\delete V' \in H\delete V'$. That is, we say that $h\delete V'$ implements the logical operator of type $h\delete (V)$, i.e.~how $h$ acts on the complement of $V'$ (see Fig.~\ref{fig:contraction_deletion_example}).

Lemma \ref{lemma:surjectivity_lemma} also allows for a straightforward proof that every stabilizer code can be thought of as a \emph{graph code} (first shown but phrased differently in \cite{schlingemann2001stabilizer, grassl2002graphs}), up to single-qubit Clifford rotations. A graph code can be defined in several different ways. One way is the operational interpretation from~\cite{hein2006entanglement}. Here, $k$ input qubits are prepared in an arbitrary state, $n$ output qubits are initialized in $\ket{+}$, and then $\CZ$ gates are performed according a given graph $G$. After the $\CZ$ gates, the input qubits are measured in the $X$ basis. This teleports the input state to the $n$ output qubits (up to some unimportant outcome-dependent Pauli corrections), effectively encoding the input state on the $n$ output qubits. This can be shown to be equivalent to implementing a code whose stabilizer group corresponds to deletion of the input qubits on the graph state $\ket{G}$~\cite{goodenough2023near}.

Lemma \ref{lemma:surjectivity_lemma} allows us to reason about stabilizer codes $H=\overbar{H}\contract V'$ in terms of the stabilizer group $\overbar{H}$ of a state, which leads to a simplified proof of the following fact.

\begin{theorem}
 Every stabilizer code is single-qubit Clifford equivalent to a graph code.
 \end{theorem}
 
 \begin{proof}
 The statement is equivalent to the claim that every stabilizer group $H$ is single-qubit Clifford equivalent to $\overbar{H}'\contract V'$ for some set $V'$ and $\overbar{H}'$ the stabilizer group of a graph state. From Lemma \ref{lemma:surjectivity_lemma} we know that there exists some stabilizer group $\overbar{H}$ of a stabilizer state (not necessarily a graph state) such that $\overbar{H}\contract V' = H$. By definition, we are free to perform single-qubit Clifford operations on the complement of $V'$. However, for $U$ an arbitrary unitary local to $V'$, $\overbar{H}\contract V' = U(\overbar{H})\contract V'$. This means that we are free to apply single-qubit Clifford operations to \emph{all} qubits of $\overbar{H}$. Using that every stabilizer state is single-qubit Clifford equivalent to a graph state~\cite{van2004graphical} the claim follows.
 \end{proof}

 The above is also an alternative proof of Corollary IV.6 in~\cite{goodenough2023near}.
 
We make one technical assumption for the remainder of this paper. We assume that deletions preserve the cardinality of the group. This is to ensure that the $k$ qubits on Alice's side can still be interpreted as the code parameter $k$, see Section V.A of~\cite{goodenough2023near} for a more detailed discussion.
 
	\subsection{Conjugate codes/subspaces}
	Any stabilizer code $H$ has several closely associated codes, which differ only by phases $\pm 1$ in front of the $h\in H$. We call these the conjugate codes/subspaces of $H$. We collect here several useful characterizations of these codes and their associated subspaces. First, not every combination of phases is possible. For a fixed generating set of size $n-k$ of $H$, all other conjugate codes are generated by the same generating set, but by selecting for each of the $n-k$ generators whether it has a minus sign in front of the generator or not. As an example, $\langle XX, ZZ\rangle$ has the four conjugate codes $\langle \pm XX, \pm ZZ\rangle$; this yields the well-known Bell basis states. Alternatively, the choice of $\pm 1$ for each generator can be interpreted as a homomorphism from $H$ to $\lbrace{-1, 1\rbrace}$. In fact, the different homomorphisms $H\mapsto \lbrace{-1, 1\rbrace}$ form all the possible \emph{characters} of $H$, i.e.~homomorphisms from $H$ to $\mathbb{C}$.

	Second, any two different such subspaces are orthogonal. This is because they differ in the phase of at least one stabilizer, meaning that any two states in two different conjugate subspaces can be deterministically distinguished by measuring that stabilizer. Note that this argument still holds for mixed states that are convex combinations of basis states of the respective subspaces.
	
	Since the $2^{n-k}$ conjugate subspaces (corresponding to the $2^{n-k}$ possible assignments of $\pm 1$ for each of the $n-k$ generators) are orthogonal and each has dimension $2^k$, the entire space is a direct sum of the conjugate subspaces. Now note that the operator $\sum_{h\in H}h$ acts as the identity on its associated subspace. What is more, the kernel of $\sum_{h\in H}h$ is given by the direct sum of all the other conjugate subspaces. This is because for any $\ket{\psi}$ in a different conjugate subspace $H'$, the two cases $h\ket{\psi}=\pm\ket{\psi}$ appear equally often\footnote{On a deeper level this follows from the orthogonality of irreducible characters, such that the above result also generalizes to prime qudit stabilizers.}. Summarizing, $\sum_{h\in H}h$ is the projector onto the stabilizer codespace $H$.

	We often need to refer to states that are normalized projectors into stabilizer codespaces, i.e.~$\frac{1}{\left|H\right|}\sum_{h\in H}h\equiv \rho_{H}$. We abuse terminology and refer to such states as code projectors, even though $\rho_{H}^2$ only equals $\rho_{H}$ up to a scalar.

 The following lemma details how different Pauli operators affect codespaces/code projectors.
	\begin{lemma}\label{lemma:orthogonal_supp2}
		Let $\rho_{H}$ be a code projector with associated stabilizer group $H$ and $P\in \mathcal{P}_n$. Then $P\left(\rho_H\right) \equiv P\rho_HP^\dagger = \rho_{P(H)} = \frac{1}{\left|H\right|}\sum_{h\in H}P(h)$ equals $\rho_{H}$ iff $P \in H^\perp$.
  If $P\not\in H^\perp$ then $\rho_{H}$ and $\rho_{P(H)}$ are distinct conjugate code projectors.
	\end{lemma}
 \begin{proof}
      If $h'\in H^\perp$, then by definition $h'$ commutes with each $h\in H$, so that $h'(\rho_{H})= \rho_{H}$. Otherwise, $h'$ anti-commutes with exactly half the elements of $H$. This leads to phases being added in front of the anti-commuting terms, yielding a conjugate code.
 \end{proof}
	In fact, $P(\rho_H)$ depends only on the coset of $H^\perp$ in $\mathcal{P}_n$ that $P$ belongs to, since two elements $h, h'$ are in the same coset iff they are the same up to multiplication by an element in $H^\perp$, yielding an equivalence relation on $\mathcal{P}_n$. We label the cosets of $H^\perp$ in $\mathcal{P}_n$ by $C_e$. Note that $H^\perp$ is a normal subgroup of $\mathcal{P}_n$. We call a generating set $E$ of $\mathcal{P}_n$ modulo $H^\perp$ an \emph{error basis}.

	\section{Results}
In this section we detail how to calculate the states $\rho_{AB}$ and $\rho_{B}$ under Pauli/depolarizing noise. First, we show that the marginal $\rho_{B}$ without noise can be interpreted in terms of deletion/contraction of the stabilizer group $H$ of $\rho_{AB}$, i.e.~we show that $\rho_{B} = \rho_{H\contract A}$\footnote{Note that the subscript on the left indicates the system on which the state lives, while on the right it indicates a code.}. Second, we study what happens when applying noise; we use the fact that $P\in \mathcal{P}_n$ maps $\rho_{H}$ to $\rho_{H'}$, where the $H'$ depends only on which coset of $H^\perp$ in $\mathcal{P}_n$ the element $P$ belongs to.   

 To start, note that without noise states $\rho_{AB}$ and $\rho_{A}/\rho_{B}$ are code projectors. This is trivially true for $\rho_{AB}$. For the marginals, we have the following result (see also~\cite{englbrecht2022transformations})
	
	\begin{prop}
		Let $\rho_{AB}$ be a pure stabilizer state with stabilizer group $\overbar{H}$. Then $\rho_{B}$ is the code projector with stabilizer group $\overbar{H}\contract A$.
	\end{prop}
	\begin{proof}
		Since $\rho_{AB}$ is a code projector, we have that $2^n\rho_{AB} = \sum_{h\in \overbar{H}}h$. Note that tracing out the $A$ system leaves only the terms that act as the identity on the $A$ system. This follows from linearity of the trace, that the trace is multiplicative over tensor products, and that $\Tr{I}=2,~\Tr{X}=\Tr{Y}=\Tr{Z}=0$. From Definition \ref{def:contraction} we find that this is the same thing as contracting the subset $A$ on $\overbar{H}$.
	\end{proof}
	The above proposition also generalizes to marginals of code projectors.

	From Lemma \ref{lemma:orthogonal_supp2} we find that applying Pauli noise to a code projector $\rho_{H}$ yields probabilistic mixtures of the different conjugate states $\rho_{H'}$\footnote{For non-trivial depolarizing each possible conjugate code has a non-zero probability mass. This need not be true for dephasing or general Pauli noise.}. Then, to fully characterize the final state, it suffices to calculate the probability $\textrm{Prob}\left(\rho_{H'}\right)$ with which each $\rho_{H'}$ appears when applying depolarizing noise to $\rho_{H}$. Let $C_e$ be the coset that takes $\rho_H$ to $\rho_{H'}$ (see Lemma~\ref{lemma:orthogonal_supp2}). We find 
	\begin{align}
		\textrm{Prob}(\rho_{H'})=&\sum_{\mathclap{\substack{P\in \mathcal{P}_n\\ \textrm{s.t. }P(\rho_H) = \rho_{H'}}}} \textrm{Prob}(P)\\
		=&\sum_{P\in C_e} \textrm{Prob}(P)\\
		=&\frac{1}{4^n}\sum_{P \in C_e} \left(1+3\lambda\right)^{n-\textrm{wt}\left(P\right)}\left(1-\lambda\right)^{\textrm{wt}\left(P\right)}\\
		=&\frac{1}{4^n}\mathcal{W}\hspace{-0.5mm}\left(C_e, 1+3\lambda, 1-\lambda\right)\label{eq:w_enum}\ .
	\end{align}
where we assume uniform depolarizing noise to express the final result in terms of weight enumerators.
	
	With the appropriate machinery in hand, it is straightforward to compute the entropy of a code projector that has undergone depolarizing noise. Since the state is a convex combination of orthogonal projectors, the associated probability distribution factorizes. That is, the probability distribution associated with an eigenstate of the noisy code projector can be thought of as first randomly selecting which one of the codespaces $H$ one belongs to (with probability $\textrm{Prob}\left(\rho_H\right)$), and then uniformly at random selecting one of the $2^k$ basis states of the code space.
	
	The associated entropy is thus given by
	\begin{align}
		\hat{S}(D(\rho_{H})) = k + S(\mathbf{s}_H)\label{eq:synd_entropy}\  ,	
	\end{align}
	where $D$ is the noise map, $S(\mathbf{s}_H)$ is the so-called \emph{syndrome entropy} of the code $H$~\cite{gallager1962low, grover2007upper, wadayama2006ensemble, matas2016analysis, li2004slepian}. That is, it is the entropy of the probability distribution associated with observing the different possible syndromes $\textbf{s}$. The syndrome entropy has been studied before in the context of classical codes, having appeared already in the original work of Gallagher on LDPC codes~\cite{gallager1962low}.
	
	From Eq.~\eqref{eq:synd_entropy} our main theorem follows.

 \begin{theorem}\label{theorem:main_theorem_all_noise}
 The coherent information of a stabilizer state across a bipartition $A, B$ that underwent uniform depolarizing noise is given by
	\begin{align}
		\hat{S}(\rho_{B}) - \hat{S}(\rho_{AB})=k + S(\mathbf{s}_{\overbar{H}\delete A }) - S(\mathbf{s}_{\overbar{H}})\ ,	
	\end{align}
	where $\overbar{H}$ is the stabilizer group of the state shared between Alice and Bob. The distributions $\mathbf{s}_{\overbar{H}\delete A}$ and $\mathbf{s}_{\overbar{H}}$ can be expressed in terms of weight enumerators.
 \end{theorem}

\kg{In the preceding discussion, we interpreted the syndrome entropy of the code $\overbar{H}\delete A$. However, $\overbar{H}\delete A$ is not a genuine stabilizer group, as it is generally non-commutative. As a result, the syndrome entropy does not correspond to the entropy of a stabilizer code that can be experimentally realized. However, $\overbar{H}\delete A$ (and its syndrome entropy) can still be interpreted as arising from a classical linear code through the well-known correspondence between stabilizer codes and binary subspaces~\cite{calderbank1998quantum}.}

In the above we assumed the setting where all qubits underwent uniform depolarizing noise. Let us now focus on the setting where only Bob's qubits undergo uniform depolarizing noise; this setting corresponds exactly to the \emph{channel setting}, i.e.~Alice prepares $k$ maximally entangled states, and then encodes and transmits one half of her state to Bob. Using the contraction/deletion framework, we straightforwardly find the following symmetric expression for the associated coherent information.

 \begin{theorem}
 The coherent information of a stabilizer state across a bipartition $A, B$ that underwent uniform depolarizing noise on system $B$ is given by
	\begin{align}
		\hat{S}(\rho_{B}) - \hat{S}(\rho_{AB})=k + S(\mathbf{s}_{\overbar{H}\delete A }) - S(\mathbf{s}_{\overbar{H}\contract A})\ ,	
	\end{align}
	where $\overbar{H}$ is the stabilizer group of the state shared between Alice and Bob.
 \end{theorem}

\begin{proof}
Clearly the state $\rho_{B}$ is the same as in the setting of Theorem \ref{theorem:main_theorem_all_noise}, so it suffices to consider only the $\hat{S}(\rho_{AB})$ term.

As in Theorem \ref{theorem:main_theorem_all_noise}, the state $\rho_{AB}$ before noise is the code projector $\rho_{\overbar{H}}$. As before, noise on only the $B$ system will yield convex mixtures of states $P\rho_{H}P^\dagger$, but where the $P$ now vanish on $A$. Let us denote the set of such Pauli strings by $\mathcal{P}^B_{n+k}  $, which has a natural bijection to $\mathcal{P}_{n+k}\contract A\cong \mathcal{P}_{n}$.

As before, an equivalence relation $\sim$ on $\mathcal{P}^B_{n+k}$ indicates the subsets of Pauli strings that affect the state $\rho_{\overbar{H}}$ in the same manner. More specifically, two elements $P_1, P_2\in \mathcal{P}^B_{n+k}$ are equivalent iff there exists $h\in \overline{H}$ such that $P_1\cdot h = P_2$. Note that since $P_1, P_2$ vanish on $A$, so does $h$. Letting $\tilde{P}_1, \tilde{P}_2$ denote the image of $P_1, P_2$ under the above bijection, we find that $P_1 \sim P_2$ is equivalent to the statement that there exists a $\tilde{h}\in \overbar{H}\contract A$ such that $\tilde{P}_1\cdot \tilde{h} = \tilde{P}_2$, i.e.~the equivalence relation corresponds to the cosets of $\overbar{H}\contract A$ in $\mathcal{P}_{n}$.
\end{proof}

	\subsection{Algorithm}
	We now sketch our algorithm for calculating the coherent information of a stabilizer state that underwent uniform depolarizing noise. The goal is to construct an error basis $E$. $E$ is then used to reconstruct the cosets $C_e$, from which one can calculate the syndrome entropy through Eq.~\eqref{eq:w_enum}.

	Let us first focus on the $\rho_{AB}$ term. Assume that a stabilizer tableau $M$ (represented as an element of $\mathbb{F}_2^{(n+k) \times 2(n+k)}$) is given for the initial state $\ket{\psi}$. This matrix has rank $n+k$ over $\mathbb{F}_2$. Our goal is to find an error basis $E$, from which the cosets $C_e$ can be reconstructed. The error basis $E$ is of size $n+k$ and each of its elements is linearly independent from $M$.
 We construct such a basis $E$ as follows. Initialize $E$ as the empty set. Start with $M$ and check whether $e_1$ is in the row span of $M$, where $e_i$ is the $i$'th standard basis vector. It suffices to check whether $M$ and $M$ with the added row $e_1$ have the same rank over $\mathbb{F}_2$. If the rank increases, add $e_1$ to $E$ and add $e_1$ as a row vector to $M$. Now repeat the procedure but with $e_2$, and repeat until the matrix has maximum rank.

	The procedure is essentially the same for the $\rho_{B}$ term with the stabilizer tableau\footnote{Often the term stabilizer tableau is reserved for commutative subgroups, which we do not restrict ourselves to here.} associated with $\overbar{H}\delete A$ replacing the one associated with $\overbar{H}$. The new stabilizer tableau is the same as the original $M$, but with the columns indexed by $A$ removed. As before, one adds basis elements to the stabilizer tableau and error basis $E$ every time it is linearly independent from the rest.

  \kg{Let us consider the example seen in Fig.~\ref{fig:contraction_deletion_example}. Two natural stabilizer tableaux of $H$ and $H\delete 1$ are given by}

  \begin{align}
  M=
  \begin{bmatrix}
1 & 0 & 0 & 0 & 1 & 1\\
0 & 1 & 0 & 1 & 0 & 1\\
0 & 0 & 1 & 1 & 1 & 0\\
  \end{bmatrix}, \quad   M'=\begin{bmatrix}
0 & 0 & 1 & 1\\
1 & 0 & 0 & 1\\
0 & 1 & 1 & 0\\
  \end{bmatrix}\ \nonumber.
  \end{align}
  
\kg{Finding an error basis for $M$ is straightforward, since distinct $Z$-type operators map graph states to distinct orthogonal states. In other words, $E=\lbrace e_{1+3}, e_{2+3}, e_{3+3}\rbrace$. The same argument does not hold for $M'$, but fortunately the first pick $e_1$ suffices, i.e.~$E=\lbrace{e_1\rbrace}$. For simplicity, we calculate only the probability distribution over cosets corresponding to $M'$. That is, the probability distribution that has entropy $S(\mathbf{s}_{H\delete 1 })$.}

\kg{Writing out the full cosets as Pauli operators yields}

\begin{align}
\begin{array}{ccc}
\textrm{Row}(M') &\quad&  e_{1}+\textrm{Row}(M')\\
\downarrow  &\quad& \downarrow\\
\begin{bmatrix}
0 & 0 & 0 & 0\\
0 & 0 & 1 & 1\\
1 & 0 & 0 & 1\\
0 & 1 & 1 & 0\\
1 & 0 & 1 & 0\\
0 & 1 & 0 & 1\\
1 & 1 & 1 & 1\\
1 & 1 & 0 & 0\\
\end{bmatrix}\  &\quad& \begin{bmatrix}
1 & 0 & 0 & 0\\
1 & 0 & 1 & 1\\
0 & 0 & 0 & 1\\
1 & 1 & 1 & 0\\
0 & 0 & 1 & 0\\
1 & 1 & 0 & 1\\
0 & 1 & 1 & 1\\
0 & 1 & 0 & 0\\
\end{bmatrix} \\
\downarrow  &\quad& \downarrow\\
\begin{bmatrix}
II & YI\\
ZZ & IY\\
XZ & YY\\
ZX & XX\\
\end{bmatrix} &\quad& \begin{bmatrix}
XI & ZI\\
YZ & XY\\
IZ & ZY\\
YX & IX\\
\end{bmatrix}\\
\downarrow  &\quad& \downarrow\\
H\delete 1 &\quad& XI \cdot \left(H\delete 1\right)\ .
\end{array}
\end{align}

\kg{By construction, the two cosets only differ in their first entries. The weight distributions $\left(\mathcal{W}_0(C_e), \mathcal{W}_1(C_e), \mathcal{W}_2(C_e)\right)$ of the two cosets are given by $(1, 2, 5)$ and $(0, 4, 4)$, respectively. Using \eqref{eq:w_enum} and the fact that there are only two cosets, we find that the entropy of the reduced state after depolarizing is given by}

\begin{align}
S\left(D\left(\rho_{H\delete 1}\right)\right) =~& k + S(\mathbf{s}_{H\delete 1})\nonumber\\
=~&1+\overbar{S}\left(\frac{4\left(1+3\lambda\right)\left(1-\lambda\right) + 4\left(1-\lambda\right)^2}{4^2}\right)\nonumber \\
=~&1+\overbar{S}\left(\frac{1-\lambda^2}{2}\right)\nonumber \ .
\end{align}

	\subsection{Reduction for dephasing on graph states}
	
	Here we limit our discussion to graph states and dephasing on them, and show how the above connections reduce to classical linear codes. 
	
	Let us define dephasing with the following parameterization, $D_Z(\cdot ) = p(\cdot)+\left(1-p\right)Z(\cdot)$.

 \begin{theorem}
 Let $\ket{G}$ be a graph state whose biadjacency matrix across a bipartition $A, B$ is given by $G_{AB}$. The coherent information of $\ket{G}$ after uniform dephasing with parameter $p$ with respect to a bipartition $A, B$ is given by
\begin{align}
S(\rho_{B}) - S(\rho_{AB}) = k+S(\mathbf{s}_{G_{AB}})-\left(n+k\right)\overbar{S}(p) \ , 
\end{align}
where $S(\mathbf{s}_{G_{AB}})$ is the syndrome entropy of the linear code with generator matrix $G_{AB}$.
\end{theorem}
\begin{proof}
 Since $\CZ$ gates commute with $Z$ gates, the entropy of $\rho_{AB}$ is given by $\left(n+k\right)\cdot S(\lbrace p, 1-p\rbrace)\equiv \left(n+k\right)\cdot \overbar{S}(p)$. Further, we are free to remove local $\CZ$ gates. We thus restrict ourselves to bipartite graphs. 
 
 To continue, we need to construct the stabilizer tableau after deletion. We find that
 \begin{align}
\left[\begin{array}{c|c|c|c}&&&\\ I_k&\:0\:&\:0\:&G_{AB}\\
&&&\\
\hline
&&&\\
\:0 \:& I_{n} & G_{AB}^T & \:0\: \\
&&&\end{array}\right] \xrightarrow{\textrm{deletion}}\left[\begin{array}{c|c}0&G_{AB}\\\hline &  \\I_{n}&\:0\:\\& \end{array}\right],\ 
 \end{align}
	where $G_{AB}$ is the biadjacency matrix across the bipartition.

 Now, only $Z$ errors occur, which act on the $G_{AB}$ part. This means that $G_{AB}$ can be interpreted as the generator matrix of a classical code. As before, we are interested in the probability distribution over the cosets/syndromes.  In this case $S(\mathbf{s})$ reduces to the syndrome entropy of the classical linear code with generator matrix $G_{AB}$.
 \end{proof}

 \begin{corollary}
 Let $\ket{G}$ be a graph state whose biadjacency matrix across a bipartition $A, B$ is given by $G_{AB}$. The coherent information of $\ket{G}$ after uniform dephasing on the $B$ system with parameter $p$ with respect to a bipartition $A, B$ is given by
\begin{align}
S(\rho_{B}) - S(\rho_{AB}) = k+S(\mathbf{s}_{G_{AB}})-n\cdot \overbar{S}(p) \ .
\end{align}
\end{corollary} 

The coherent information for dephasing noise can be calculated using a similar technique as for depolarizing noise.

	We close with two comments. First, assume the graph state in question has been created by first preparing $\ket{+}^{\otimes n}$ and then applying $\CZ$ gates. We can include noisy preparation of the $\ket{+}^{\otimes n}$, where we can include both depolarizing and dephasing noise. This follows from the fact that depolarizing acts as dephasing on $\ket{+}^{\otimes n}$, and that $Z$ gates commute with $\CZ$ gates.
	Second, we can also include depolarizing noise after the $\CZ$ gates has been applied. Let $Q_1(s)$ and $Q_2(s)$ be the probabilities for seeing a syndrome $s$ for dephasing and depolarizing noise separately. Then the probability $Q_3(s)$ of seeing a syndrome $s$ for the combined models is given by 
	\begin{align}
		Q_3(s) = \sum_{\mathclap{\substack{s_1,s_2\\ \textrm{s.t. }s_1\oplus s_2 = s}}}Q_1\left(s_1\right)\cdot Q_2\left(s_2\right) \ .
	\end{align}
	That is, interpreting $Q_1$, $Q_2$ as functions on $\mathbb{Z}_2^{n-k}$, $Q_3$ is the convolution of $Q_1$ and $Q_2$ with respect to the group $\mathbb{Z}_2^{n-k}$. One can also apply the associated Fourier transform (the Hadamard Transform) to the $Q_1$ and $Q_2$, multiply them pointwise, and perform the inverse Fourier transform~\cite{pommerening2005fourier}. We do not report any numerics for the case of both dephasing and depolarization.

	\section{Numerical results}
	We plot the coherent information for several noisy stabilizer states in Fig.~\ref{fig:compdepol}. We consider here the graph states associated with a single Bell pair, the repetition code on seven qubits, the 5-qubit code, the $\left[4,2,2\right]$ code and the $\left[8, 3, 3\right]$ code (under uniform noise on all qubits). As noted before, several choices of graph states lead to equivalent projectors, but the syndrome entropy of the total graph state may differ. We show the two graphs that we consider for the $\left[4,2,2\right]$ code in Fig.~\ref{fig:422}.
Finally, for the $\left[8, 3, 3\right]$ code~\cite{gottesman1996class} we use the graph state from~\cite{cafaro2014scheme}).

 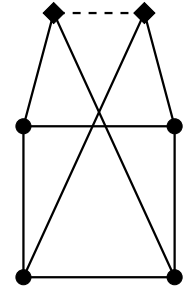
\begin{figure}[h!]
\centerfloat
\input{graphexample.tex}
\caption{The two graphs used associated with the $\left[4, 2, 2\right]$ code, where they only differ in the dashed line. Both states have the same marginal on the non-diamond qubits, but differ in the complement marginal.}
\label{fig:422}
\end{figure}

We observe that better codes lead to more robust states. Note that the two states associated with the $\left[4, 2, 2\right]$ code have the same reduced entropies but differ in their syndrome entropy of the $\left[n+k, 0, d\right]$ code associated with the total graph state.
	
\begin{figure}
\centerfloat
\hspace{2mm}
	\includegraphics[clip,  width=0.50\textwidth, trim = 17mm 0mm 0mm 0mm]{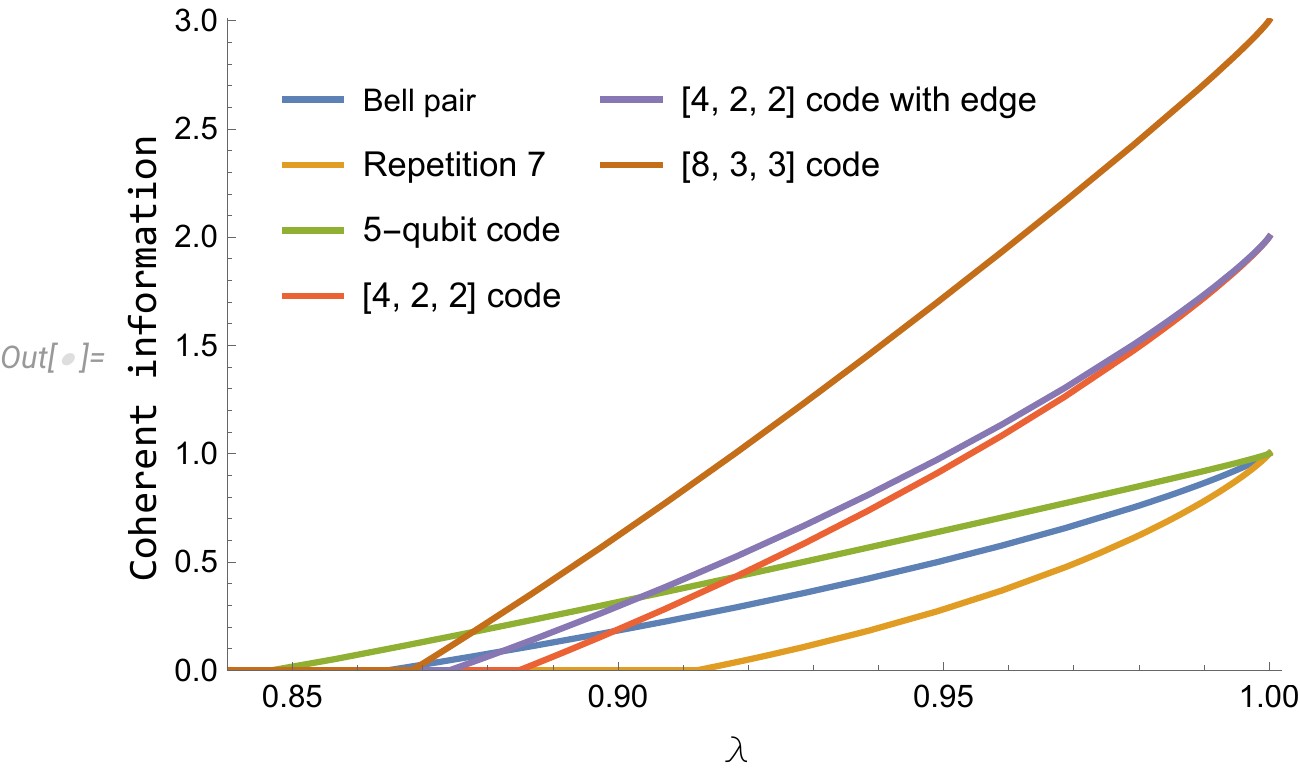}
	\vspace*{-3.3mm}	\caption{Comparison of the coherent information for several different graph states under depolarizing noise.}
	\label{fig:compdepol}
\end{figure}

In Fig.~\ref{fig:compdeph} we compare the coherent information for several different graph states under dephasing noise. We see that the repetition codes can achieve non-zero coherent information for larger values of dephasing parameter, $p$ (due to the fact that the distance increases), similar to the results from~\cite{sajjad2024lower}. On the other hand the coherent information is small even for little dephasing (i.e.~$p$ large), since repetition codes have $k=1$.

\begin{figure}
\centerfloat
	\includegraphics[clip,  width=0.5\textwidth, trim = 11mm 0mm 0mm 0mm]{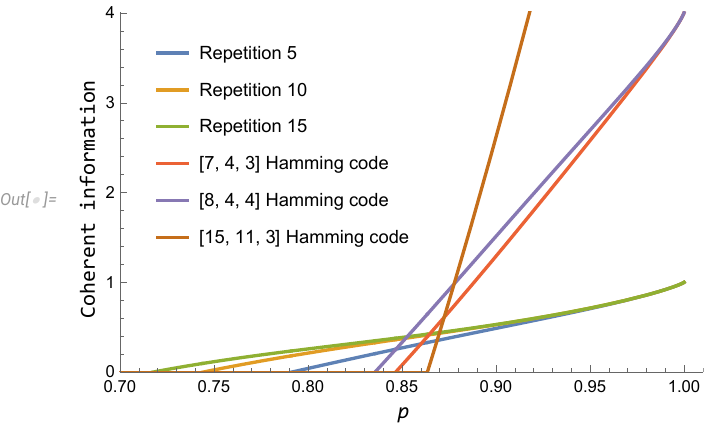}
	\vspace*{-2.3mm}
	\caption{Comparison of the coherent information for several different graph states under dephasing noise.}
	\label{fig:compdeph}
\end{figure}

We see that the robustness of stabilizer states is correlated with how good the associated codes are, for both dephasing and depolarizing noise.

	\section{Discussion}
	In this paper we interpreted noisy stabilizer states in terms of certain associated stabilizer codes. This suggests that stabilizer states robust against noise/distribution in networks have good codes associated to them; we confirmed this intuition for a number of codes. \kg{However, there seems no structural way to find the state with the largest coherent information for fixed $n$, $k$ and noise parameter $\lambda$. This should come as no surprise, since finding optimal codes for any metric is a difficult task.}

 \kg{On a higher level}, we have shown that the spectra of noisy stabilizer states can be expressed in terms of weight enumerators, quantities that are known to be hard to compute (even in the classical case) but have been well-studied~\cite{hill1986first, rains1998quantum, chauhan2021hamming}.
	Fortunately bounds on the syndrome entropy have been investigated for the classical case; it would be of interest to see if those results can be exploited for the case considered here. This would however require a good enough control on the error terms, since the coherent information involves the difference of two syndrome entropies.

The techniques developed here apply straightforwardly to the multipartite setting. One possible application would be to express upper bounds~\cite{goodenough2016assessing, bennett2014quantum} on the multipartite squashed entanglement~\cite{yang2009squashed} to find upper bounds on the multipartite distillable key.
	
\section{Acknowledgements}
The authors would like to extend thanks to Dion Gijswijt for discussions on contractions, deletions and related topics.
This work was funded by the Army Research Office
(ARO) MURI on Quantum Network Science under grant
number W911NF2110325.

	
	\bibliographystyle{IEEEtran}
	\bibliography{references}

\appendix

\section{Proof of Lemma \ref{lemma:surjectivity}}\label{sec:surjectivity}
We give here a proof of Lemma \ref{lemma:surjectivity_lemma}. It is more convenient to work with the dual $H^\perp$ instead of $H$ itself. Lemma \ref{lemma:surjectivity_lemma} can then be rephrased as follows.

 \begin{lemma}\label{lemma:surjectivity_lemma2}
 For every dual $H^\perp$ of a stabilizer group $H$ on a set of qubits $V$, there exists a stabilizer group $\overbar{H}$ of a stabilizer state on a set of qubits $V' \cup V$ such that $H^\perp = \overbar{H}\delete V'$.
 \end{lemma}
 \begin{proof}

Note that $H^\perp$ can always be chosen to have a generating set of the form $\lbrace x_{1}, x_2, \ldots,x_{k}, x_{k+1}, \ldots, x_{n}, z_{1},z_{2}\ldots, z_{k}\rbrace$, such that $\lbrace x_{k+1}, x_{k+2}, \ldots, x_{n}\rbrace$ generates $H$. Here, the labels are chosen such that two elements anti-commute iff they are of the form $x_i$ and $z_j$ and $i=j$. One way of seeing this is by choosing a `canonical' stabilizer group~\cite{jansen2022enumerating} $F = \langle Z_{k+1}, Z_{k+2}, \ldots, Z_{n}\rangle$, noting that $F^\perp = \langle Z_1, Z_2, \ldots, Z_{k}, Z_{k+1}, \ldots Z_{n}, X_{1},X_{2}, \ldots, X_{k}\rangle$, and that the symplectic group acts transitively on the set of all stabilizer groups. An example is given by the groups in Fig.~\ref{fig:contraction_deletion_example}, where $H\contract 1 = \langle YY\rangle = \langle x_2\rangle$ and $\left(H\contract 1\right)^\perp = H\delete 1 = \langle XX, YY, ZX\rangle = \langle x_1, x_2, z_1\rangle$.

We now construct a generating set for $\overbar{H}$ from $\lbrace x_{1}, x_2, \ldots,x_{k}, x_{k+1}, \ldots, x_{n}, z_{1},z_{2}\ldots, z_{k}\rbrace$, by mapping each element $x_{i}$ ($z_{i}$) to some $\overbar{x}_i$ ($\overbar{z}_{i}$). Since $H=\langle x_{k+1}, x_{k+2}, \ldots, x_n \rangle$, we require that $\overbar{x}_{i}$ is the same as $x_i$, except with identity operators on the $k$ new qubits. For our example, we see that $\overbar{x}_2 = IYY$. 

For $\overbar{x}_i$ with $1\leq i \leq k$, we attach a single $X$ operator to the $i$'th new entry and identity everywhere else. We do the same for $\overbar{z}_i$ and $1\leq i \leq k$, but where we attach a single $Z$ operator to the $i$'th new entry instead of an $X$.
The new $\overbar{x}_{i}$ and $\overbar{z}_{i}$ are independent and commute pairwise. Furthermore, the group generated by the $\overbar{x}_i$ and $\overbar{z}_i$ has size $2^{n+k}$, and each of its elements has length $n+k$. This implies that the group generated by the $\overbar{x}_i$ and $\overbar{z}_i$ forms the stabilizer group of a state. Finally, by construction it is clear they generate a group that under deletion yields $H$.

For our example, we find that $\overbar{x}_1 = XXX$ and $\overbar{z}_1 = ZZX$, which yields the stabilizer group of the graph state $\ket{K_3}$. Note that a different choice of $x_1, x_2, z_1$ would have yielded a different stabilizer group $H$.
\end{proof}

\end{document}

%% file: contraction_deletion.tex
\def\xshift{2}
\begin{tikzpicture}
\node[circle, fill=black, draw, scale=0.6] (v2) at (-1-1,+1-2){};
\node[circle, fill=black, draw, scale=0.6] (v3) at (+1-1, +1-2){};

\node[diamond, fill=black, draw, scale=0.6] (v1) at (0.0-1, 2.5-2){};
\node[scale=1] (1) at (0.3-1, 2.8-2){$1$};
\node[scale=1] (2) at (-1-0.3-1, +1+0.3-2){$2$};
\node[scale=1] (3) at (+1+0.3-1, +1+0.3-2){$3$};

\draw[line width = 0.3mm] (v1) -- (v2) -- (v3) -- (v1);

\node at (3, 1.7) [rectangle,minimum size=2cm] {$\overline{H} = $ \begin{tabular}{|c c|}
    \hline
    $III$ & $YYI$ \\
    $XZZ$ & $YIY$\\
    $ZXZ$ & $IYY$\\
    $ZZX$ & $XXX$\\
    \hline
\end{tabular}};

\node at (3.5, -0.5) [rectangle,minimum size=2cm] {$\overline{H}\backslash 1 = $ \begin{tabular}{|c c|}
    \hline
    $I_L = II$ & $Y_L=YI$ \\
    $X_L = ZZ$ & $Y_L =IY$\\
    $Z_L = XZ$ & $I_L = YY$\\
    $Z_L = ZX$ & $X_L = XX$\\
    \hline
\end{tabular}};

\node at (2.5, -2) [rectangle,minimum size=2cm] {$\overline{H}/1 = $ \begin{tabular}{|c c|}
    \hline
    $II$ & $YY$ \\
    \hline
\end{tabular}};

\end{tikzpicture}

%% file: graphexample.tex

\begin{tikzpicture}

\node[circle, fill=black, draw, scale=0.6] (a) at (-1,+1){};
\node[circle, fill=black, draw, scale=0.6] (b) at (+1, +1){};
\node[circle, fill=black, draw, scale=0.6] (c) at (+1, -1){};
\node[circle, fill=black, draw, scale=0.6] (d) at (-1, -1){};

\node[diamond, fill=black, draw, scale=0.6] (d1) at (0.6, 2.5){};
\node[diamond, fill=black, draw, scale=0.6] (d2) at (-0.6, 2.5){};


\draw[line width = 0.3mm] (1,1) -- (1,-1) -- (-1,-1) -- (-1,+1) -- cycle;

\draw[line width = 0.3mm] (-0.6, 2.5) -- (1,-1);
\draw[line width = 0.3mm] (-0.6, 2.5) -- (-1,1);

\draw[line width = 0.3mm] (0.6, 2.5) -- (-1,-1);
\draw[line width = 0.3mm] (0.6, 2.5) -- (1,1);

\draw[line width = 0.3mm, dashed] (0.6, 2.5) --  (-0.6, 2.5);




\end{tikzpicture}